\documentclass[%
 reprint,
 amsmath,amssymb,
 aps,
 prb,
,colorlinks=true,linkcolor=apsblue, citecolor=apsblue, urlcolor=apsblue]{revtex4-2}

\usepackage{graphicx}% Include figure files
\usepackage{dcolumn}% Align table columns on decimal point
\usepackage{bm}% bold math
\usepackage{amsthm}% theorem environments
\usepackage{mathtools}
\usepackage{booktabs}
\usepackage{hyperref}
\usepackage{cleveref}
\usepackage{xcolor}
\definecolor{apsblue}{rgb}{0,0,0.8}

\crefname{table}{table}{tables}
\Crefname{table}{Table}{Tables}
\crefname{figure}{figure}{figures}
\Crefname{figure}{Figure}{Figures}
\crefname{equation}{Eq.}{Eqs.}
\Crefname{equation}{Equation}{Equations}
\crefname{section}{Sec.}{Secs.}
\Crefname{section}{Section}{Sections}
\crefname{appendix}{Appendix}{Appendices}
\Crefname{appendix}{Appendix}{Appendices}

\newcommand{\bdag}[1]{\hat b_{#1}^{\dagger}}
\newcommand{\bop}[1]{\hat b_{#1}}
\newcommand{\nop}[1]{\hat n_{#1}}
\newcommand{\ket}[1]{|#1\rangle}
\newcommand{\bra}[1]{\langle #1|}

\newcommand{\Tr}{\operatorname{Tr}}
\newcommand{\Hi}{\mathcal{H}}
\newcommand{\Fo}{\mathcal{F}}
\newcommand{\bd}{\partial A}
\newcommand{\Bin}{\operatorname{Bin}}

\theoremstyle{plain}
\newtheorem{theorem}{Theorem}

\theoremstyle{definition}

\theoremstyle{remark}

\newcommand{\fullfigref}[2]{\hyperref[#1]{\ref*{#1}(#2)}}

\begin{document}

\preprint{APS/123-QED}
%\nocite{apsrev42Control}
\title{Odd/Even or Half?  Entanglement Anomaly in the Bose-Hubbard model}

\author{Humberto E. Hern\'andez-L\'opez}%
 \email{humbertoe@estudiantes.fisica.unam.mx}
 
 \author{Santiago F. Caballero-Benitez}
\email{scaballero@fisica.unam.mx}

\affiliation{%
 Departamento de F\'\i sica Cu\'antica y Fot\'onica, LSCSC-LANMAC, Instituto de F\'\i sica, Universidad Nacional Aut\'onoma de M\'exico, Ciudad de M\'exico 04510, Mexico
}%

%\date{\today}
\setcitestyle{super}
\makeatletter
\renewcommand\@biblabel[1]{\textsuperscript{#1.}}
\makeatother

\begin{abstract}
The area law relates the bipartite entanglement entropy of a quantum many-body ground state to the size of the boundary between the subsystems, but the geometry of this boundary is rarely discussed. We inspect this in the 1D Bose-Hubbard model at fixed density by comparing four spatial bipartitions of the periodic lattice: first half, second half, even sites, and odd sites; sharing the same number of sites but differing in how the boundary is arranged. 
We find  analytical limits with perturbation theory: in the Mott insulator the contiguous cut obeys the area lay while the alternating cut obeys a volume law $S\propto N_s$, in this sense an anomaly,  and for the superfluid both cuts colapse to the binomial saturation due to delocalization of the state. 
We formulate these limits as a statement about the many-body problem using a generalized slave-boson approach based on mean-field with quantum fluctuations while verifying with Exact Diagonalization (ED) for small lattice sizes and Densitiy Matrix Renormalization Group (DMRG) simulations for $N_s\gg 1$. The slave-boson Gaussian ground state allows to compute the entanglement entropy from a reduced correlation matrix for any desired bipartition consistent with ED and DMRG results. 
Using slave bosons the  computational cost is set by the local cutoff $n_{\max}$ rather than the Hilbert space dimension, so we can reach lattice sizes far beyond ED. 
Our method is capable of establishing the partition-dependent scaling laws as a many-body feature, not only a finite-size effect, in great agreement with the ED for $N_s\in[4,10]$ and DMRG for larger lattice sizes.
\end{abstract}

\maketitle

The entanglement entropy (EE) has become the standard tool for characterizing quantum phase transitions \cite{Amico2008,Eisert2010,Laflorencie2016}. A central result in this context is the area law: for gapped local Hamiltonians, the ground-state entropy scales with the boundary area $\partial A$ of the subsystem \cite{Eisert2010,Hastings2007}. In 1D, this implies $S = \mathcal{O}(1)$ in a gapped phase, with a logarithmic correction $S \sim (c/3)\ln L$ at criticality \cite{CalabreseCardy2004,Vidal2003}. The Bose-Hubbard model \cite{Fisher1989,Jaksch1998,Lewbook}, experimentally realized with ultracold atoms in optical lattices \cite{Greiner2002}, is the paradigmatic system to study such transitions. Its EE across the Mott insulator-superfluid (MI-SF) transition has been studied via slave-boson mean-field theory \cite{Frerot2016}, quantum Monte Carlo simulations \cite{CasianoDiaz2023}, and measured experimentally \cite{Islam2015}. Entanglement has likewise proven a sharp diagnostic of order and competition in related lattice-boson settings with long-range or cavity-mediated interactions \cite{Sharma2022,LozanoMendez2022,CaballeroMekhov2015,RamirezBarajas2025}.

Virtually all of these studies employ `contiguous' bipartitions, in which subsystem $A$ is a connected block of sites. The area law, however, concerns the size of the boundary between subsystems, not its geometry. This motivates a natural question: What happens when the boundary is made extensive, as in the alternating (odd/even) bipartition, where every nearest-neighbor bond is cut? Studies on quantum spin chains suggest the EE should then scale with the number of broken bonds \cite{Keating2006}, replacing the area law with a volume law $S\propto N_s$, as demonstrated for free bosonic chains \cite{Rossignoli2011}, and therefore anomalous. However, no systematic study has addressed this across the MI-SF transition of the Bose-Hubbard model, where strong interactions and quantum correlations in the many-body quantum state can be dominant.

We compare four spatial bipartitions of a 1D ring: first half, second half, odd sites, and even sites. Using group-theory arguments and the structure of tensor products \cite{Zanardi2001,Zanardi2004}, we show that translational invariance alone does not force the half and alternating cuts to yield the same entanglement entropy. Analytic perturbation theory reveals that in the MI regime the half cut obeys the area law while the alternating cut follows a volume law, whereas in the SF regime both cuts converge through delocalization.

The methodological core of this work, and what allows us to promote these analytic limits into a statement about the full many-body problem, is a generalized slave-boson (SB) mean field supported by Exact diagonalization\cite{ZhangDong2010,Lew} (ED) and Density Matrix Renormalization Group\cite{DMRG1,DMRG2}  (DMRG) computational simulations using ITensor\cite{ITensor1,ITensor2,ITensor3}. The SB approach is built on a two-sublattice Gutzwiller state \cite{Sheshadri1993,vanOosten2001} dressed with Bogoliubov-de Gennes fluctuations, it produces a Gaussian ground state whose entanglement entropy is fixed by a reduced correlation matrix \cite{Peschel2003,PeschelEisler2009} that can be restricted to any subset of sites (contiguous or alternating) within a single formalism. The slave-boson approach applied to the SF-MI transition was introduced by Fr\'erot and Roscilde \cite{Frerot2016}; here we modify it to a two-sublattice geometry so that both cut types are accessible at once. Crucially, its cost is set by the local cutoff $n_\mathrm{max}$ and the number of momenta $k$ rather than by the Hilbert-space dimension $D$. Therefore, it is possible to reach lattice sizes far beyond  ED and shows that the partition-dependent scaling laws persist and extend cleanly into the many-body regime, i.e. that this is not a finite-size effect. This is confirmed by both ED small lattices and DMRG for larger lattices, the three methods are consistent when computationally feasible, certifying that the SB extrapolates the correct scaling laws and shows when the perturbation treatment needs further corrections.

\emph{Model and central question.}  The Bose-Hubbard model \cite{Fisher1989,Jaksch1998} for $N_s$ sites with $N$ particles, with periodic boundary conditions, is
\begin{equation}
\mathcal{H} = -t\sum_{\langle i,j\rangle}\left(\bdag{i}\bop{j}+\mathrm{H.c.}\right)
        +\frac{U}{2}\sum_{i=1}^{N_s}\nop{i}(\nop{i}-1),
\label{eq:BH}
\end{equation}
where $\langle i,j\rangle$ runs over the $N_s$ site bonds, i.e., the pairs $(i,i+1)$ with $i+N_s\equiv i$. This is a cycle graph $C_{N_s}$ (\cref{fig:grafo}). Given a partition of the sites $\Lambda=\{1,\dots,N_s\}$ into two subsets $A$ and $B=\Lambda\setminus A$, the bipartite entanglement entropy for the ground state $\ket{\psi_0}$ is
\begin{equation}
S_A=-\Tr\!\left[\rho^A\ln\rho^A\right],\qquad
\rho^A=\Tr_B\ket{\psi_0}\bra{\psi_0}.
\label{eq:SA}
\end{equation}

\begin{figure}[t!]
\centering
\includegraphics[width=0.3\linewidth]{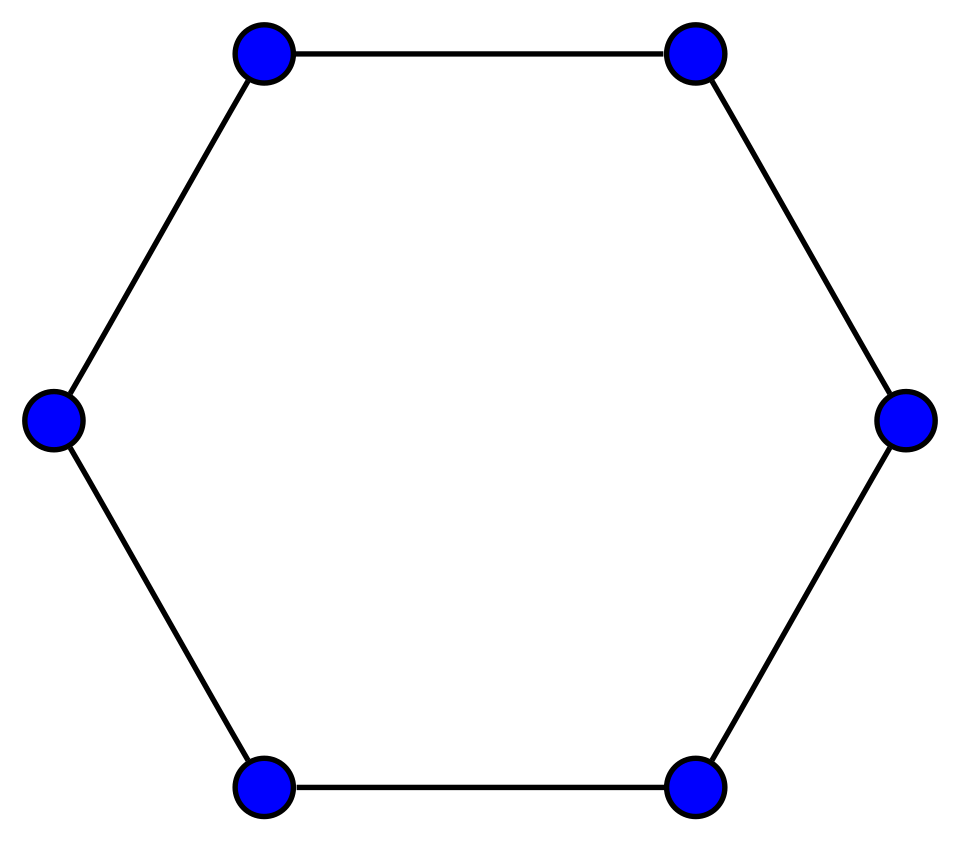}
\caption{Example of cycle graph $C_6$.}
\label{fig:grafo}
\end{figure}
The question to investigate is the following. The Hamiltonian \cref{eq:BH} is translationally invariant and invariant under reflections of the ring; should $S_A$ therefore be independent of which sites are grouped into $A$, as long as $|A|$ does not change? In particular, should the entropy coincide for the subsets
\[
A_{\mathrm{half}}=\{1,\dots,N_s/2\}
\quad\text{and}\quad
A_{\mathrm{odd}}=\{1,3,5,\dots,N_s-1\},
\]
both with $|A|=N_s/2$ sites? Counterintuitive, the answer we found is that $S_A$ is not only a function of the state, but also of the structure of the tensor product. The symmetry of the Hamiltonian applies to the state, but it only relates those bipartitions that the same symmetry permutes. The bipartition must be analyzed independently of the symmetry \cite{Zanardi2001,Zanardi2004}, since it has a correspondence to the tensor product structure.

Since we work in the canonical ensemble (fixed particle number), the physical subspace is $\Hi_N=\{\ket{\psi}\in\Fo:\hat N\ket{\psi}=N\ket{\psi}\}$, of dimension $D=\binom{N+N_s-1}{N}$. For any partition $\Lambda=A\sqcup B$ the Fock space factorizes as $\Fo=\Fo_A\otimes\Fo_B$, and $\Hi_N$ respects the fixed total number,
\begin{equation}
\Hi_N=\bigoplus_{n=0}^{N}\Fo_A^{(n)}\otimes\Fo_B^{(N-n)},
\label{eq:sumadirecta}
\end{equation}
so that $\rho^A$ is block-diagonal in the particle number in $A$, a consequence of the number superselection that governs entanglement between particle-conserving modes \cite{Wiseman2003}. Writing $\rho^A=\bigoplus_n p_n\,\rho^A_n$ with $\Tr\rho^A_n=1$ gives the exact decomposition
\begin{equation}
S_A=\underbrace{-\sum_n p_n\ln p_n}_{S_{\mathrm{num}}}
   +\underbrace{\sum_n p_n S(\rho^A_n)}_{S_{\mathrm{conf}}} ,
\label{eq:num_conf}
\end{equation}
where the number term $S_{\mathrm{num}}$ is set by the distribution of particles across the cut (links entanglement and particle-number fluctuations \cite{Klich2009,Song2010}) and dominates in both the MI and SF limits, while the configurational term $S_{\mathrm{conf}}$ is maximal in the transition region. The model \cref{eq:BH} at commensurate filling has two exact limits, MI at $t/U\to0$ and SF at $t/U\to\infty$, whose 1D transition occurs at $(t/U)_c=0.3050(10)$ \cite{Ejima2011,Kuhner2000}. We show that in the MI the geometry of the bipartition determines $S_A$ and how it scales, while in the SF the geometry is irrelevant and $S_{\mathrm{half}}$, $S_{\mathrm{odd}}$ coincide.

\emph{Symmetry arguments.} 
Two exact statements hold for the full range of $t/U$; their proofs are given in the SM \cref{app:proofs}.

First, complementarity: for any pure bipartite state, $\rho^A$ and $\rho^B$ share the same nonzero spectrum, so $S_A=S_B$ (Schmidt symmetry, \cref{th:schmidt}). This immediately yields
\[
S_{\{1,\dots,N_s/2\}}=S_{\{N_s/2+1,\dots,N_s\}},
\qquad
S_{\mathrm{odd}}=S_{\mathrm{even}} .
\]
It does not, however, necessarily determine whether $S_{\mathrm{half}}=S_{\mathrm{even}}$.

Second, invariance under lattice symmetries: if a site permutation $\pi$ satisfies $[V_\pi,\mathcal{H}]=0$ and $\ket{\psi_0}$ is an eigenvector of $V_\pi$, then $S_{\pi(A)}=S_A$ ( SM \cref{th:invariancia}). The eigenvector hypothesis holds only when $\pi$ is an automorphism of the cycle $C_{N_s}$.

The key observation is that no automorphism maps the half cut to the alternating one. The automorphism group of $C_{N_s}$ is the dihedral group $D_{2N_s}$, and every element preserves adjacency, hence maps a subgraph to an isomorphic one. But $A_{\mathrm{half}}=\{1,\dots,N_s/2\}$ has $N_s/2-1\ge 2$ internal edges, whereas $A_{\mathrm{odd}}=\{1,3,\dots,N_s-1\}$ has none (two odd sites are never neighbors for even $N_s$). A graph with edges cannot be isomorphic to one without, so for even $N_s$ no $\pi$ obeys $\pi(A_{\mathrm{half}})=A_{\mathrm{even/odd}}$. Translational invariance guarantees $S_A=S_{T(A)}$, but it does not permit rearranging which sites belong to $A$, since that changes the factorization $\Fo=\Fo_A\otimes\Fo_B$ with respect to which the entropy is defined \cite{Zanardi2001,Zanardi2004}. The intuition that both bipartitions must agree confuses a symmetry of the state with a freedom of the observer. This argument only rules out a guaranteed equality; whether the values coincide must be decided regime by regime.

\emph{Analytic limits: Mott insulator} ($t/U\ll 1$).  Split \cref{eq:BH} as $\mathcal{H}=\mathcal{H}_U+\mathcal{H}_t$. For $t=0$ and $N=N_s$ the ground state is the non-degenerate MI state $\ket{\psi_{MI}}=\ket{1,1,\dots,1}$ with $E_0=0$, and the gap to the doublon-hole degenerate states is $\Delta=U$. Perturbation theory at first-order gives
\begin{equation}
\ket{\psi_0}=\ket{\psi_{MI}}+\varepsilon\!\!\sum_{\langle i,j\rangle}\!\left(\ket{D_{ij}}+\ket{D_{ji}}\right)+O(\varepsilon^{2}),
\quad
\varepsilon\equiv\frac{\sqrt2\,t}{U},
\label{eq:psi1}
\end{equation}
where $\ket{D_{ij}}$ is the doublon-hole state ($n_i=2$, $n_j=0$, unit occupancy elsewhere) and the sum runs over the $2N_s$ oriented bonds. Only terms whose doublon-hole crosses the boundary $\bd$ (the set of bonds with $i\in A$ and $j\in B$ or viceversa) contribute to entanglement; the rest are reabsorbed into a product state \cite{Alba2012}.

Collecting the boundary terms into a coefficient matrix $C$ and diagonalizing $CC^{\mathsf T}$ (details in SM \cref{app:mott}), one finds a dominant eigenvalue and $r$ small ones $\lambda_k=\varepsilon^2 c_k[1+O(\varepsilon)]$ with $\sum_k c_k=\Tr(CC^{\mathsf T})=2|\bd|$. The resulting entropy is
\begin{equation}
S_A=\varepsilon^{2}\Big[\,2|\bd|\Big(1+\ln\tfrac{1}{\varepsilon^{2}}\Big)-\sum_{k}c_k\ln c_k\Big]+O\!\left(\varepsilon^{3}\right),
\label{eq:master}
\end{equation}
whose first term is manifestly the area law \cite{Eisert2010,Hastings2007}. Evaluating \cref{eq:master} for the three cuts:
\begin{align}
S_{\mathrm{half}}&=4\varepsilon^{2}\Big(1+\ln\tfrac1{\varepsilon^{2}}\Big),
\label{eq:Smitad}\\
S_{\mathrm{even/odd}}&=\varepsilon^{2}\!\left[2N_s\Big(1+\ln\tfrac1{\varepsilon^{2}}\Big)
-2\!\sum_{k=0}^{m-1}\!f\!\Big(\tfrac{\pi k}{m}\Big)\right],
\label{eq:Salt}\\
S_{\mathrm{site}}&=4\varepsilon^{2}\Big(1+\ln\tfrac1{\varepsilon^{2}}\Big)-4\ln 2\,\varepsilon^{2},
\label{eq:Ssitio}
\end{align}
with $m\equiv N_s/2$, $f(x)=4\cos^{2}\!x\,\ln(4\cos^{2}\!x)$, and $0\ln0\equiv0$. The half cut has $|\bd|=2$ and does not depend on $N_s$ (area law); the alternating cut has $|\bd|=N_s$ and scales $\propto N_s$ (volume law). The single-site cut, isomorphic to the half-cut boundary, matches $S_{\mathrm{half}}$ up to a constant.

\emph{Analitic limits: Superfluid} ($t/U\gg 1$). 
For $U=0$ the exact ground state is the single-mode condensate $\ket{\psi_0}=(N!)^{-1/2}(\bdag{0})^{N}\ket{0}$ with $\bdag{0}=N_s^{-1/2}\sum_j\bdag{j}$. Introducing sublattice modes $a_A^\dagger,a_B^\dagger$ and using the binomial theorem, the Schmidt decomposition is (see the SM \cref{app:sf})
\begin{equation}
\ket{\psi_0}=\sum_{n=0}^{N}\sqrt{\binom{N}{n}p^{n}q^{N-n}}\;\ket{n}_A\ket{N-n}_B,
\label{eq:schmidtSF}
\end{equation}
with $p=\tfrac{\ell}{N_s},\ q=1-p$, so that $S_A=H(\Bin(N,p))$ (see the EM and the SM) depends only on $|A|=\ell$, not on its geometry. Hence $S_{\mathrm{half}}=S_{\mathrm{even/odd}}=H(\Bin(N,\tfrac12))$: not as consequence of the lattice symmetry, but because of the delocalization of the state. For large $N$ at unit filling,
\begin{equation}
S_A\simeq\tfrac12\ln\!\Big(\tfrac{\pi e}{2}N_s\Big),
\label{eq:logley}
\end{equation}
the logarithmic law governing the SF phase \cite{Klich2009,Song2010,Giamarchi2003}.

\emph{Scaling laws.}  \Cref{tab:leyes} summarizes the two limits. In the MI, the cut geometry matters greatly: the even-odd cut has an extensive boundary, while the half-half cut always has two boundary points regardless of $N_s$ \cite{Hastings2007,Eisert2010}, the same phenomenon reported for spin chains with comb-type bipartitions \cite{Keating2006}. In the SF, the geometry becomes irrelevant, since the state occupies the lattice macroscopically and entanglement is fixed by the binomial distribution of the particles.

\begin{table}[h!]
\centering
\caption{Behavior of $S_A$ with $N_s$ (analytical limits).}
\label{tab:leyes}
\begin{ruledtabular}
\begin{tabular}{lll}
 & MI ($t/U\ll 1$) & SF ($t/U\gg 1$)\\
\colrule
$|\bd|$ half-half   & $2$          & $2$\\
$|\bd|$ even-odd  & $N_s$        & $N_s$\\
scaling $S_{\mathrm{half}}$ & const. & $\tfrac12\ln\!\Big(\tfrac{\pi e}{2}N_s\Big)$\\
scaling $S_{\mathrm{even/odd}}$   & $\propto N_s$ & $\tfrac12\ln\!\Big(\tfrac{\pi e}{2}N_s\Big)$ \\
\end{tabular}
\end{ruledtabular}
\end{table}

\emph{Slave-boson treatment of an arbitrary bipartition.}  The analytic limits in the SF an MI are exact but live at the edges of the phase diagram and, in the MI, only to dominant order in $\varepsilon$. To decide if these scaling laws are genuine many-body features rather than small-size or leading-order accidents, we need a method that treats the contiguous and alternating cuts on exactly the same footing, and reaches lattice sizes beyond what exact diagonalization allows. The slave-boson (SB) approach fullfils both requirements. \cite{Sheshadri1993,Frerot2016,Sharma2022}. Here we present the construction; the explicit assembly of the quadratic Hamiltonian and the dominant order analysis are shown in the SM \cref{app:sb}.

\begin{figure*}[t]
\centering
\includegraphics[width=0.92\textwidth]{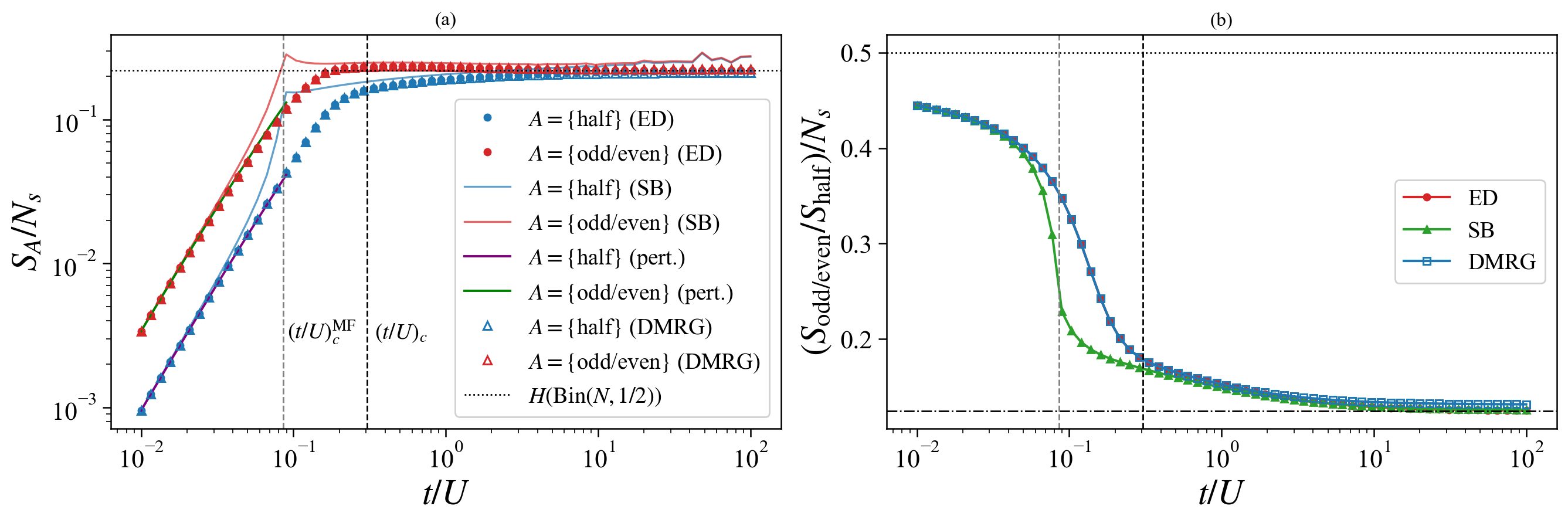}
\caption{Full range of $t/U$, $N_s=10$. Left: $S_A$ from ED, SB, and perturbation theory \cref{eq:master}, with the binomial saturation of the SF. Right: ratio $S_{\mathrm{even/odd}}/S_{\mathrm{half}}$: it starts near $N_s/2$ in the MI, decays across the critical region, and tends to $1$ in the SF. The critical points $(t/U)_c=0.305$ and $(t/U)_c^{\mathrm{MF}}\approx0.086$ are marked.}
\label{fig:cmp_completo}
\end{figure*}

\begin{figure*}[t]
\centering
\includegraphics[width=0.9\textwidth]{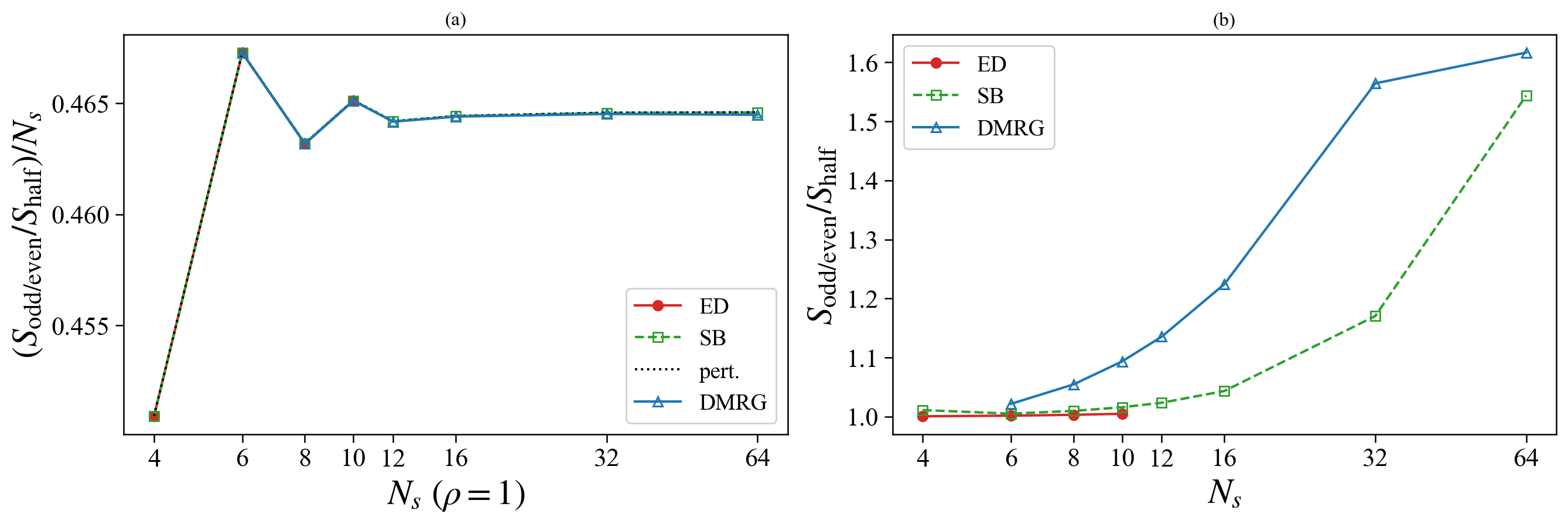}
\caption{Scaling with $N_s$ at unit filling, comparing ED and SB. Left, $t/U=10^{-2}$ (MI). Right, $t/U=10^{2}$ (SF).}
\label{fig:cmp_esc}
\end{figure*}

\emph{Mean-field state and Gaussian fluctuations.}  The starting point is the two-sublattice (odd/even) mean-field state
\begin{equation}
\ket{\Phi_0}=\bigotimes_{i\in E}\ket{\alpha_e}\otimes\bigotimes_{i\in O}\ket{\alpha_o},
\qquad
\ket{\alpha_{e,o}}=\sum_{n=0}^{n_{\max}}\alpha^{(e,o)}_n\ket{n},
\label{eq:sb-ansatz-main}
\end{equation}
whose local amplitudes minimize the energy per site at fixed mean density $\rho=N/N_s$, and normalized for each sublattice. Fluctuations around \cref{eq:sb-ansatz-main} are treated within BdG theory: on each sublattice one builds a local orthonormal basis whose first vector is the optimized state $\ket{\alpha}$, and the remaining $n_{\max}$ excited states play the role of fluctuation bosons $\hat\gamma_{r,\alpha}$. Folding the Brillouin zone by the staggered structure ($k$ coupled to $k+\pi$, with $f_k=2\cos k$ in 1D) casts the quadratic Hamiltonian into the standard BdG block form
\begin{equation}
\mathcal{H}^{(2)}=\tfrac14\sum_{k}\hat\Gamma_k^\dagger\,
\mathcal H_{\mathrm{BdG}}(k)\,\hat\Gamma_k,
\qquad
\mathcal H_{\mathrm{BdG}}(k)=\begin{pmatrix}A_k&B_k\\ B_k^{*}&A_k^{*}\end{pmatrix},
\label{eq:sb-Hbdg-main}
\end{equation}
diagonalized by a paraunitary (Bogoliubov) transformation $T_k$ with the bosonic metric $\sigma=\mathrm{diag}(\mathbb I,-\mathbb I)$. The same transformation delivers the ground-state correlation matrix $\mathcal C(k)$,  see the End Matter and the SM (\cref{app:sb-bdg}) .

\begin{table}[h!]
\centering
\caption{Ratio $S_{\mathrm{even/odd}}/S_{\mathrm{half}}$ in the deep MI limit ($t/U=10^{-3}$, $\rho=1$).}
\label{tab:cmp_ratio}
\begin{ruledtabular}
\begin{tabular}{rccccc}
$N_s$ & ED & SB & DMRG & PT & $N_s/2$\\
\colrule
$4$  & $1.8036$ & $1.8036$  & --        & $1.8037$  & $2$\\
$6$  & $2.8036$ & $2.8036$  & $2.8036$  & $2.8037$  & $3$\\
$8$  & $3.7054$ & $3.7054$  & $3.7054$  & $3.7055$  & $4$\\
$10$ & $4.6512$ & $4.6512$  & $4.6511$  & $4.6513$  & $5$\\
$12$ & --       & $5.5702$  & $5.5700$  & $5.5703$  & $6$\\
$16$ & --       & $7.4307$  & $7.4303$  & $7.4308$  & $8$\\
$32$ & --       & $14.8661$ & $14.8645$ & $14.8664$ & $16$\\
$64$ & --       & $29.7334$ & $29.7267$ & $29.7339$ & $32$\\
\end{tabular}
\end{ruledtabular}
\end{table}

\emph{Computational Simulations.} We now compare the SB against exact diagonalization (ED). In ED, $\mathcal{H}$ is built in sparse form \cite{ZhangDong2010,Lew}, the ground state is obtained by Lanczos, and $\rho^A$ is diagonalized as the squared singular values of the amplitude matrix $M_{ab}$ (see SM \cref{eq:rho-a-schmidt}); it is numerically exact but restricted to $N_s\le10$ by the growth of $D$. The SB, through \cref{eq:S_gauss_main}, supplies the same entropies at a cost independent of $D$ and for both cuts on equal footing, so the two methods overlap on $N_s\le10$ and the SB continues alone beyond it. Moreover, the results are confirmed using DMRG, considering periodic boundary conditions, while being sub-optimal due to entanglement growth as it is well know\cite{DMRG2}. The periodic boundary conditions make simulations with DMRG difficult, but we can achieve larger lattice sizes to analyze the scaling with respect to ED. We perform our simulations with a combination of CPU/GPU implementation of the algorithm with ITensor\cite{ITensor1,ITensor2,ITensor3}..

\Cref{fig:cmp_completo} is the central result over the full range of $t/U$. Where ED is available the SB tracks it curve for curve, and the perturbative master formula \cref{eq:master} matches both deep in the MI; the clearest signature of the geometric effect is the ratio $S_{\mathrm{even/odd}}/S_{\mathrm{half}}$, which starts near $N_s/2$ in the deep MI (the volume-law versus area-law contrast), decays across the critical region, and tends to $1$ in the SF. In the MI phase all of $S_A$ comes from doublon-hole fluctuations, precisely those generated by the perturbation theory of the analytic treatment shown and known to control the entanglement structure of the insulating phase \cite{Lauchli2008,Alba2012}, and the SB reproduces the ratio to four significant figures (\cref{tab:cmp_ratio}). The convergence in the SF is not a lattice symmetry but the delocalization of the state erasing the cut geometry; both methods saturate near $H(\Bin(N_s,\tfrac12))$.'

\Cref{fig:cmp_esc} isolates what only the SB can show: the scaling with $N_s$ at fixed $t/U$, followed well past the ED ceiling. Deep in the MI the contiguous cut stays flat (area law) while the alternating cut grows linearly (volume law), and the two collapse onto the common binomial value in the SF. The scaling laws it certifies deep in each phase are exactly those confirmed by ED.

In this Letter, we have observed the effect of the choice of bipartition on the entanglement entropy and how strong interactions can modify significantly the overall picture leading to a volumetric law. The Schmidt symmetry theorem ensures that tracing over $A$ or $B$ is equivalent, and it was verified numerically that $S_{\mathrm{half}_1}=S_{\mathrm{half}_2}$ and $S_{\mathrm{even}}=S_{\mathrm{odd}}$. However, translational invariance does not require $S_{\mathrm{half}}=S_{\mathrm{even}}$: no automorphism of the ring maps $A_{\mathrm{half}}\to A_{\mathrm{even/odd}}$, since the subgraphs are not isomorphic. The extent to which one bipartition differs from the other was computed analytically in the perturbative MI regime, and we showed why they converge again in the SF, not by invariance but by delocalization. The single-site cut, isomorphic to the half-half boundary, provides an additional consistency check. The tool used to test these statements is the slave-boson mean field with quantum fluctuations. Its Gaussian structure lets a single reduced correlation matrix deliver the entropy of the contiguous and the alternating cut, and its cost lets it reach large values of $N_s$, inaccesible for ED due to Hilbert space growth. Our results are consistent with DMRG simulations as the system size increases providing an independent verification of this asymmetry.  We find  that in general $S_{\mathrm{even/odd}}/S_{\mathrm{half}}\to N_s/2$ for MI, and the collapse of both cuts through delocalization in the SF. Its leading-order reduction coincides mode for mode with the perturbative spectrum,  ED and DMRG simulations. We expect the same approach to be a convenient route to partition-dependent entanglement in other quantum matter settings\cite{caballero2026quantum}, as well as other systems where the interplay of strong quantum correlations and interaction driven effects might arise with unusual behavior~\cite{Kitaev2006,Levin2006,Savary2016,Giamarchi2008,Li2008,Schuch2011,Schlawin2019}.
\\

This work is partially supported by the grants UNAM-DGAPA-PAPIIT: IN118823, UNAM-DGPA-PAPIIT: IG101826, UNAM-CIC: Apoyo a Laboratorios Nacionales 2025, CONAHCYT/SECIHITI: LNC-2023-51,  as well as by grants from NVIDIA and utilized NVIDIA RTX 6000 Ada.

% Bibliography --------------------------------

%\bibliographystyle{apsrev4-2}
%\bibliography{bib_ea}% Produces the bibliography via BibTeX.

\clearpage
\section*{End Matter}
\emph{Entropy of any subset from one correlation matrix.}
The essential point is that the fluctuation ground state is Gaussian, hence fully determined by its correlations. Transforming $\mathcal C(k)$ to real space through the unfolding relation
\begin{equation}
\hat\gamma_x=\frac{1}{\sqrt N}\sum_k e^{ikx}\big(\hat\gamma_k+(-1)^x\hat\gamma_{k+\pi}\big),
\label{eq:sb-unfold-main}
\end{equation}
and simply keeping the rows and columns labelled by the sites $x\in A$, one obtains the reduced correlation matrix $\mathcal C_A$. The eigenvalues $\{n_m\}$ of $\sigma_A\mathcal C_A$ are the entanglement-mode occupancies, and the entropy is \cite{Peschel2003,PeschelEisler2009,Casini2009}
\begin{equation}
S_A=\sum_m\big[(1+n_m)\ln(1+n_m)-n_m\ln n_m\big].
\label{eq:S_gauss_main}
\end{equation}
Because \cref{eq:sb-unfold-main,eq:S_gauss_main} require nothing but the list of sites in $A$, the half-half and even/odd cuts are computed within the same formalism, at the same finite $N_s$; the geometry of the cut enters only through which rows of $\mathcal C_A$ are retained. This allows to reach large $N_s$ and is what makes the SB the natural tool for the central question of this paper.

\emph{Consistency with perturbation theory.} 
Each SB calculation is performed on a finite periodic lattice of $N_s$ sites at fixed unit density with a discrete Brillouin zone $k\in\{2\pi m/N_s\}$; it is therefore not a thermodynamic-limit calculation, and the entropy is evaluated for that specific $N_s$. Its cost, however, is controlled by the local cutoff $n_{\max}$ and the number of momenta rather than by $D=\binom{2N_s-1}{N_s}$, so it accesses periodic lattice sizes far beyond the $N_s\le10$ ceiling of ED. Evaluating the partition-dependent scaling at large $N_s$ then turns the analytic limits for $t/U\gg1$ and $t/U\ll 1$ into a statement about the many-body regime.

That the SB and the perturbation theory describe the same physics in the MI is made explicit by expanding \cref{eq:S_gauss_main} at small occupancies $n_m\ll1$, where it reduces to
\begin{equation}
S_A\simeq-\sum_m n_m\ln n_m+\sum_m n_m,
\label{eq:sb-leading-main}
\end{equation}
which has exactly the form of the master formula \cref{eq:master} under the identification $n_m\equiv\lambda_m=\varepsilon^2 c_m$. The entanglement-mode occupancies of the SB are, mode for mode, the eigenvalues $\{c_k\}$ of the boundary matrix $CC^{\mathsf T}$ computed analytically in \cref{eq:Smitad,eq:Salt,eq:Ssitio}, so the two constructions reconstruct the same spectrum (\cref{app:sb-mi}). An important clarification about the SB, being a mean-field, is that it places the MI-SF critical point at $(t/U)_c^{\mathrm{MF}}=\tfrac12(3-2\sqrt2)\approx0.086$ \cite{Sheshadri1993,vanOosten2001}, while the exact 1D value is $(t/U)_c=0.305$ \cite{Ejima2011,Kuhner2000}, so on the critical region the SB is only qualitative; for the rest of the regime, it is quantitative.

\emph{Entropy of entanglement in the perfect SF ($U=0$)}. For details on the derivation, see the SM. 
\begin{eqnarray}
S_A&=&-\sum_{n=0}^{N}\binom{N}{n}p^{n}q^{N-n}
\ln\!\Big[\binom{N}{n}p^{n}q^{N-n}\Big]
\\
&=&H\!\big(\Bin(N,p)\big).
\end{eqnarray}

\clearpage
\appendix

\widetext
\begin{center}
\textbf{\large Supplemental material for: Odd/Even or Half?  Entanglement Anomaly in the Bose-Hubbard model}
\\
$\phantom{a}$
\\
{Humberto E. Hern\'andez-L\'opez$^{1}$ and Santiago F. Caballero-Benitez$^{1}$}
 %\email{scaballero@fisica.unam.mx}
 \\
{%
$^{1}$Instituto de Física, LSCSC-LANMAC, Universidad Nacional Autónoma de México, Ciudad de México 04510, Mexico
}%
\end{center}
%\end{\widetext}
%%%%%%%%%% Merge with supplemental materials %%%%%%%%%%
%%%%%%%%%% Prefix a "S" to all equations, figures, tables and reset the counter %%%%%%%%%%
\setcounter{equation}{0}
\setcounter{figure}{0}
\setcounter{table}{0}
\setcounter{page}{1}
\makeatletter
\renewcommand{\theequation}{S\arabic{equation}}
\renewcommand{\thefigure}{S\arabic{figure}}

\section{Proofs of the exact results}
\label{app:proofs}

\begin{theorem}[Schmidt symmetry]
Let $\ket{\psi}\in\Fo_A\otimes\Fo_B$ be a pure, normalized state. Then $\rho^A=\Tr_B\ket{\psi}\bra{\psi}$ and $\rho^B=\Tr_A\ket{\psi}\bra{\psi}$ have the same spectrum of non-zero eigenvalues, with the same multiplicities. Consequently $S_A=S_B$. This is the standard Schmidt-decomposition symmetry \cite{NielsenChuang}; we include the proof to fix notation used below.
\label{th:schmidt}
\end{theorem}

\begin{proof}
Considering orthonormal bases $\{\ket{\alpha}_A\}$ and $\{\ket{\beta}_B\}$, write $\ket{\psi}=\sum_{\alpha\beta}M_{\alpha\beta}\ket{\alpha}_A\ket{\beta}_B$ with $\sum_{\alpha\beta}|M_{\alpha\beta}|^2=1$. Computing the partial traces element by element,
\begin{equation}
\rho^A_{\alpha\alpha'}=\sum_\beta M_{\alpha\beta}M_{\alpha'\beta}^*,
\label{eq:rho-a-schmidt}
\end{equation}
so $\rho^A=MM^{\dagger}$ and $\rho^B=(M^{\dagger}M)^{\mathsf T}$. With the SVD $M=W\Sigma V^{\dagger}$, $MM^{\dagger}=W\Sigma\Sigma^{\mathsf T}W^{\dagger}$ and $M^{\dagger}M=V\Sigma^{\mathsf T}\Sigma V^{\dagger}$; $\Sigma\Sigma^{\mathsf T}$ and $\Sigma^{\mathsf T}\Sigma$ differ only in trailing zeros, so their nonzero entries are $\{\sigma_k^2\}$. The transposition does not affect the spectrum, and since $S(\rho)$ depends only on the spectrum with $0\ln0\equiv0$, we conclude $S_A=S_B$.
\end{proof}

\begin{theorem}[Invariance of $S_A$ under permutation]
Let $\pi$ be a permutation of sites with $[V_\pi,\mathcal{H}]=0$, and let $\ket{\psi}$ be an eigenvector of $V_\pi$, $V_\pi\ket{\psi}=e^{i\theta}\ket{\psi}$. Then $S_{\pi(A)}=S_A$ for every $A\subseteq\Lambda$.
\label{th:invariancia}
\end{theorem}

\begin{proof}
Since $\pi$ is a bijection on sites, $V_\pi$ maps $\Fo_A\otimes\Fo_B$ to $\Fo_{\pi(A)}\otimes\Fo_{\pi(B)}$ and factorizes as $V_\pi=V_\pi^{A}\otimes V_\pi^{B}$ with $V_\pi^{A}:\Fo_A\to\Fo_{\pi(A)}$ unitary. Writing $\ket{\psi'}=V_\pi\ket{\psi}$,
\begin{align}
\rho^{\pi(A)}[\psi']&=\Tr_{\pi(B)}\!\left[(V^A_\pi\otimes V^B_\pi)\ket{\psi}\bra{\psi}(V^A_\pi\otimes V^B_\pi)^{\dagger}\right]\nonumber\\
&=V^A_\pi\,\rho^{A}[\psi]\,V^{A\dagger}_\pi.
\end{align}
Hence $\rho^{\pi(A)}[\psi']$ and $\rho^{A}[\psi]$ are unitarily equivalent, and $S_{\pi(A)}[\psi]=S_A[\psi]$.
\end{proof}

\section{Details of the MI perturbation theory}
\label{app:mott}

This appendix expands the derivation summarized in the main text. We work at unit filling ($N=N_s$) and split the Hamiltonian \cref{eq:BH} as
$\mathcal{H}=\mathcal{H}_U+\mathcal{H}_t$, with
\[
\mathcal{H}_U=\frac{U}{2}\sum_i\nop{i}(\nop{i}-1),
\qquad
\mathcal{H}_t=-t\sum_{\langle i,j\rangle}\big(\bdag{i}\bop{j}+\mathrm{H.c.}\big).
\]

\subsection{Ground state to first order}

For $t=0$ the ground state of $\mathcal{H}_U$ is the Mott state $\ket{\psi_{MI}}=\ket{1,1,\dots,1}$ with $E_0=0$. It is non-degenerate: any other configuration with $\sum_i n_i=N_s$ must place two particles on some site, giving energy $\ge U>0$. The first excited manifold consists of the doublon-hole states $\ket{D_{ij}}$, defined by $n_i=2$, $n_j=0$, and unit occupancy on every other site; each has energy $\tfrac{U}{2}\,[2(2-1)]=U$, so the gap is $\Delta=U$.

Acting with $\mathcal{H}_t$ on $\ket{\psi_{MI}}$, a single hop along the oriented bond $j\to i$ produces
\begin{equation}
\bdag{i}\bop{j}\ket{\psi_{MI}}
=\sqrt{n_j}\,\sqrt{n_i+1}\,\ket{D_{ij}}
=\sqrt{2}\,\ket{D_{ij}},
\end{equation}
since $n_j=1$ and $n_i+1=2$ on the Mott background. Hence the only nonzero matrix element is $\bra{D_{ij}}\mathcal{H}_t\ket{\psi_{MI}}=-\sqrt2\,t$. Standard first-order perturbation theory, $\ket{\psi_0}=\ket{\psi_{MI}}+\sum_m \frac{\bra{m}\mathcal{H}_t\ket{\psi_{MI}}}{E_0-E_m}\ket{m}$, with $E_0-E_m=-U$ for every doublon-hole state, then gives
\begin{equation}
\ket{\psi_0}=\ket{\psi_{MI}}
+\varepsilon\!\!\sum_{\langle i,j\rangle}\!\big(\ket{D_{ij}}+\ket{D_{ji}}\big)
+O(\varepsilon^{2}),
\qquad
\varepsilon\equiv\frac{\sqrt2\,t}{U},
\label{eq:psi1app}
\end{equation}
where both orientations $D_{ij}$ and $D_{ji}$ of each bond appear, for a total of $2N_s$ terms.

\subsection{Only boundary terms entangle}

Fix a bipartition $A|B$ and recall that the boundary $\bd$ is the set of bonds $\langle i,j\rangle$ with one endpoint in $A$ and the other in $B$. Writing each term of \cref{eq:psi1app} as a product $\ket{a}_A\ket{b}_B$, the doublon-hole excitation of a given bond falls into exactly one of three classes:
\begin{enumerate}
\item[(i)] both endpoints in $A$: the term is
$\ket{\chi_A}\ket{\psi_{MI,B}}$, with the $B$ factor unchanged;
\item[(ii)] both endpoints in $B$: the term is
$\ket{\psi_{MI,A}}\ket{\chi_B}$, with the $A$ factor unchanged;
\item[(iii)] the bond lies in $\bd$: the term is $\ket{a_\mu}\ket{b_\mu}$
with $\ket{a_\mu}\perp\ket{\psi_{MI,A}}$ and $\ket{b_\mu}\perp\ket{\psi_{MI,B}}$.
\end{enumerate}
Collecting the interior contributions (i) and (ii) with the leading term,
\begin{align}
&\ket{\psi_{MI,A}}\ket{\psi_{MI,B}}
+\varepsilon\ket{\chi_A}\ket{\psi_{MI,B}}
+\varepsilon\ket{\psi_{MI,A}}\ket{\chi_B}\nonumber\\
&\qquad=\ket{\psi_{MI,A}'}\ket{\psi_{MI,B}'}+O(\varepsilon^{2}),
\\
&\ket{\psi_{MI,X}'}\equiv\ket{\psi_{MI,X}}+\varepsilon\ket{\chi_X},
\end{align}
which is a product state and therefore carries zero entanglement entropy. Only the boundary terms (iii) survive as genuine entanglement to this order \cite{Alba2012}.

\subsection{The boundary matrix $C$}

Order the relevant Fock states of $A$ as the rows and those of $B$ as the columns, and define the boundary matrix
\begin{equation}
C_{ab}=
\begin{cases}
1, & (a,b)\ \text{comes from an oriented boundary bond},\\
0, & \text{otherwise.}
\end{cases}
\end{equation}
Together with the (rank-one) leading term, the full amplitude matrix of the state is $M=M_0+\varepsilon C+O(\varepsilon^2)$, with $M_0=e_{\psi_{MI,A}}e_{\psi_{MI,B}}^{\mathsf T}$.

The reduced density matrix $\rho^A=MM^{\dagger}$ has eigenvalues equal to the squared singular values of $M$. At zeroth order $M_0$ has a single unit singular value along $e_{\Omega_A},e_{\Omega_B}$. Projecting off that direction with $P_A=\mathbb I-e_{\Omega_A}e_{\Omega_A}^{\mathsf T}$ (and $P_B$ analogously), the singular-value perturbation theorem \cite{HornJohnson} states that the remaining (small) singular values of $M$ coincide, to leading order, with those of
\[
P_A M P_B=\varepsilon\,P_A C P_B+O(\varepsilon^{2})=\varepsilon\,C+O(\varepsilon^2),
\]
the last equality holding because, by construction of class (iii), every row and column of $C$ is already orthogonal to $e_{\Omega_A}$ and $e_{\Omega_B}$. Squaring, the small eigenvalues of $\rho^A$ are
\begin{equation}
\lambda_k=\varepsilon^{2}c_k\,[1+O(\varepsilon)],\qquad k=1,\dots,r,
\label{eq:eigen-smallapp}
\end{equation}
where $c_1,\dots,c_r$ are the nonzero eigenvalues of $CC^{\mathsf T}$. Normalization fixes the dominant eigenvalue,
\begin{eqnarray}
\lambda_0&=&1-\varepsilon^{2}\sum_k c_k+O(\varepsilon^3),
\\
\sum_k c_k&=&\Tr(CC^{\mathsf T})=\sum_{a,b}C_{ab}^{2}=2|\bd|,
\end{eqnarray}
the last count being simply the number of oriented boundary bonds.

\subsection{From the spectrum to $S_A$}

We now assemble $S_A=-\lambda_0\ln\lambda_0-\sum_{k\ge1}\lambda_k\ln\lambda_k$ term by term. For the dominant eigenvalue, using $\lambda_0=1-2\varepsilon^2|\bd|$ and $\ln(1-x)=-x+O(x^2)$,
\begin{equation}
-\lambda_0\ln\lambda_0=2\varepsilon^{2}|\bd|+O(\varepsilon^{4}).
\end{equation}
For the small eigenvalues \cref{eq:eigen-smallapp},
\begin{align}
-\sum_{k\ge1}\lambda_k\ln\lambda_k
&=-\varepsilon^{2}\sum_k c_k\big[\ln\varepsilon^{2}+\ln c_k\big]+O(\varepsilon^{3})\nonumber\\
&=2\varepsilon^{2}|\bd|\,\ln\tfrac{1}{\varepsilon^{2}}
-\varepsilon^{2}\sum_k c_k\ln c_k+O(\varepsilon^{3}),
\end{align}
where we used $\sum_k c_k=2|\bd|$ in the first term. Adding the two contributions gives the master formula
\begin{equation}
S_A=\varepsilon^{2}\Big[\,2|\bd|\Big(1+\ln\tfrac{1}{\varepsilon^{2}}\Big)
-\sum_{k}c_k\ln c_k\Big]+O\!\left(\varepsilon^{3}\right).
\label{eq:masterapp}
\end{equation} 

The first term is proportional to the number of cut bonds $|\bd|$ and is the area law \cite{Eisert2010,Hastings2007}; the second, spectrum-dependent term is a subleading correction. Everything about the geometry of the cut enters through the single matrix $C$.

\subsection{Evaluating \cref{eq:masterapp} for each cut}

We now build $C$ explicitly. It is enough to track, for each oriented boundary bond, which orthogonal $A$-state it excites and which $B$-state accompanies it. We denote by $d_i$ a doublon at site $i$ and by $h_i$ a hole at site $i$.

\paragraph{Half-half.}
Take $A=\{1,\dots,N_s/2\}$. The boundary consists of the two bonds $(\tfrac{N_s}{2},\tfrac{N_s}{2}+1)$ and $(N_s,1)$, so $|\bd|=2$ and there are four class-(iii) terms. Their $A$-states $\{d_{N_s/2},h_{N_s/2},d_1,h_1\}$ are all distinct, and each pairs with a unique $B$-state $\{h_{N_s/2+1},d_{N_s/2+1},h_{N_s},d_{N_s}\}$ dictated by particle conservation. The pairing is therefore a bijection and $C$ is a permutation matrix. Ordering the four $A$-states as $(d_{N_s/2},h_{N_s/2},d_1,h_1)$ and the $B$-states as $(h_{N_s/2+1},d_{N_s/2+1},h_{N_s},d_{N_s})$,
\begin{equation}
C_{\mathrm{half}}=
\begin{pmatrix}
1&0&0&0\\
0&1&0&0\\
0&0&1&0\\
0&0&0&1
\end{pmatrix},
\qquad
C_{\mathrm{half}}C_{\mathrm{half}}^{\mathsf T}=\mathbb I_4 .
\end{equation}
Thus $c=(1,1,1,1)$ and $\sum_k c_k\ln c_k=0$, so \cref{eq:masterapp} gives
\begin{equation}
S_{\mathrm{half}}=4\varepsilon^{2}\Big(1+\ln\tfrac1{\varepsilon^{2}}\Big),
\label{eq:Smitadapp}
\end{equation}
which is independent of $N_s$, always has two bonds, for any $N_s$.

\paragraph{Even-odd.}
Take $A=\{1,3,\dots,N_s-1\}$ and write $m\equiv N_s/2$. Now every one of the $N_s$ bonds crosses the cut, so $|\bd|=N_s$ and there are $2N_s$ class-(iii) terms; the pairing is no longer bijective, because each doublon site in $A$ has 2 neighbors in $B$. Consider first the block in which the doublon sits in $A$ (a hole then sits on a neighboring $B$-site). Labelling the $A$-rows by $d_1,d_3,\dots,d_{2m-1}$ and the $B$-columns by $h_2,h_4,\dots,h_{2m}$, the doublon $d_{2r-1}$ connects to the holes on its two neighbors, giving a $1$ on the diagonal and a $1$ on the cyclic super-diagonal:
\begin{equation}
C^{(\mathrm I)}=\mathbb I_m+P,
\qquad
P=
\begin{pmatrix}
0&1&&&\\
&0&1&&\\
&&\ddots&\ddots&\\
&&&0&1\\
1&&&&0
\end{pmatrix},
\end{equation}
with $P$ the $m\times m$ cyclic permutation. As a circulant matrix, $C^{(\mathrm I)}$ is diagonalized by the Fourier modes \cite{HornJohnson} and has eigenvalues $1+\omega^{k}$, $\omega=e^{2\pi i/m}$, so
\begin{equation}
\operatorname{spec}\!\big(C^{(\mathrm I)}C^{(\mathrm I)\mathsf T}\big)
=\big\{\,|1+\omega^{k}|^{2}\,\big\}_{k=0}^{m-1}
=\Big\{\,4\cos^{2}\tfrac{\pi k}{m}\,\Big\}_{k=0}^{m-1}.
\end{equation}
The complementary block, with the doublon in $B$ and the hole in $A$, is identical in structure and gives the same spectrum. Since $d_i$ and $h_i$ are orthogonal Fock states, the two blocks do not mix and $CC^{\mathsf T}$ is block-diagonal; its spectrum is the union of the two, i.e.\ each value $4\cos^2(\pi k/m)$ with multiplicity two. As a check,
\[
\Tr(CC^{\mathsf T})=2\sum_{k=0}^{m-1}4\cos^{2}\tfrac{\pi k}{m}
=2\cdot 4\cdot\tfrac{m}{2}=4m=2N_s=2|\bd|,
\]
using $\sum_{k=0}^{m-1}\cos^2(\pi k/m)=m/2$. Substituting the doubled spectrum into \cref{eq:masterapp},
\begin{equation}
S_{\mathrm{even/odd}}=\varepsilon^{2}\!\left[2N_s\Big(1+\ln\tfrac1{\varepsilon^{2}}\Big)
-2\!\sum_{k=0}^{m-1}\!f\!\Big(\tfrac{\pi k}{m}\Big)\right],
\label{eq:Saltapp}
\end{equation}
with $f(x)=4\cos^{2}\!x\,\ln\!\big(4\cos^{2}\!x\big)$ and the convention $0\ln0=0$ (the mode $k=m/2$, where $\cos(\pi/2)=0$, contributes nothing). The leading term now scales as $2N_s$: the boundary is extensive and $S_{\mathrm{even/odd}}\propto N_s$, a volume law.

It is instructive to see both matrices for the smallest ring, $N_s=4$ ($m=2$). For the half cut $A=\{1,2\}$ one has, in the ordering above, $C_{\mathrm{half}}=\mathbb I_4$. For the even-odd cut $A=\{1,3\}$ the doublon-in-$A$ block is
\[
C^{(\mathrm I)}=\mathbb I_2+P
=\begin{pmatrix}1&1\\1&1\end{pmatrix},
\qquad
C^{(\mathrm I)}C^{(\mathrm I)\mathsf T}
=\begin{pmatrix}2&2\\2&2\end{pmatrix},
\]
with eigenvalues $\{4,0\}=\{4\cos^2 0,\,4\cos^2\tfrac{\pi}{2}\}$, exactly as predicted; together with the identical doublon-in-$B$ block the full spectrum is $\{4,4,0,0\}$.

\paragraph{Single site.}
For $A=\{1\}$ one again has $|\bd|=2$, but now the two boundary bonds $(1,2)$ and $(N_s,1)$ share the site $1$. The class-(iii) terms are: a doublon at $1$ paired with a hole at either neighbor ($2$ or $N_s$), and a hole at $1$ paired with a doublon at either neighbor. Ordering the two $A$-rows as $(d_1,h_1)$ and the four $B$-columns as $(h_2,h_{N_s},d_2,d_{N_s})$,
\begin{equation}
C_{\mathrm{site}}=
\begin{pmatrix}
1&1&0&0\\
0&0&1&1
\end{pmatrix},
\qquad
C_{\mathrm{site}}C_{\mathrm{site}}^{\mathsf T}
=\begin{pmatrix}2&0\\0&2\end{pmatrix}
=\operatorname{diag}(2,2).
\end{equation}
Hence $c=(2,2)$, $\sum_k c_k=4=2|\bd|$, and $\sum_k c_k\ln c_k=4\ln2$, so
\begin{equation}
S_{\mathrm{site}}=4\varepsilon^{2}\Big(1+\ln\tfrac1{\varepsilon^{2}}\Big)-4\ln 2\,\varepsilon^{2}.
\label{eq:Ssitioapp}
\end{equation}
The area-law term matches $S_{\mathrm{half}}$, both have $|\bd|=2$, and the two cuts differ only through the constant $-4\ln2\,\varepsilon^2$ coming from the shared vertex, consistent with the two boundaries being isomorphic as subgraphs.

\section{Details of the SF limit}
\label{app:sf}

This appendix derives the results of the main text in full. At $U=0$ the Hamiltonian is diagonalized by the zero-momentum single-particle mode
\begin{equation}
\bdag{0}=\frac{1}{\sqrt{N_s}}\sum_{j=1}^{N_s}\bdag{j},
\end{equation}
and the $N$-particle ground state is the condensate
\begin{equation}
\ket{\psi_0}=\frac{1}{\sqrt{N!}}\big(\bdag{0}\big)^{N}\ket{0},
\qquad E_0=-2tN.
\label{eq:becapp}
\end{equation}

\subsection{Schmidt decomposition}

Fix a bipartition $A|B$ with $|A|=\ell$. Define the two normalized sublattice creation operators
\begin{equation}
a_A^{\dagger}=\frac{1}{\sqrt{\ell}}\sum_{i\in A}\bdag{i},
\qquad
a_B^{\dagger}=\frac{1}{\sqrt{N_s-\ell}}\sum_{j\in B}\bdag{j},
\end{equation}
which inherit bosonic commutators $[a_A,a_A^{\dagger}]=[a_B,a_B^{\dagger}]=1$ and, since $A$ and $B$ are disjoint, $[a_A,a_B^{\dagger}]=0$. The condensate mode is a weighted sum of the two,
\begin{equation}
\bdag{0}=\sqrt{p}\,a_A^{\dagger}+\sqrt{q}\,a_B^{\dagger},
\qquad
p=\frac{\ell}{N_s},\quad q=1-p,
\end{equation}
the weights being simply the fraction of sites on each side. Raising to the $N$-th power with the binomial theorem (the two operators commute),
\begin{equation}
\big(\bdag{0}\big)^{N}
=\sum_{n=0}^{N}\binom{N}{n}\,p^{n/2}q^{(N-n)/2}\,
\big(a_A^{\dagger}\big)^{n}\big(a_B^{\dagger}\big)^{N-n}.
\end{equation}
Acting on $\ket{0}=\ket{0}_A\ket{0}_B$ and using $\big(a_X^{\dagger}\big)^{k}\ket{0}_X=\sqrt{k!}\,\ket{k}_X$,
\begin{equation}
\ket{\psi_0}
=\sum_{n=0}^{N}\frac{\binom{N}{n}\sqrt{n!\,(N-n)!}}{\sqrt{N!}}\,
p^{n/2}q^{(N-n)/2}\,\ket{n}_A\ket{N-n}_B .
\end{equation}
The prefactor simplifies, since $\binom{N}{n}\sqrt{n!(N-n)!}/\sqrt{N!}=\sqrt{\binom{N}{n}}$, leaving the Schmidt form
\begin{equation}
\ket{\psi_0}=\sum_{n=0}^{N}\sqrt{\binom{N}{n}p^{n}q^{N-n}}\;\ket{n}_A\ket{N-n}_B .
\label{eq:schmidtSFapp}
\end{equation}
The states $\{\ket{n}_A\}$ and $\{\ket{N-n}_B\}$ are orthonormal, so \cref{eq:schmidtSFapp} is already a valid Schmidt decomposition, and the Schmidt coefficients are the square roots of the binomial probabilities $\binom{N}{n}p^{n}q^{N-n}$. The reduced density matrix is therefore diagonal with these probabilities, and
\begin{eqnarray}
S_A&=&-\sum_{n=0}^{N}\binom{N}{n}p^{n}q^{N-n}
\ln\!\Big[\binom{N}{n}p^{n}q^{N-n}\Big]
\\
&=&H\!\big(\Bin(N,p)\big).
\end{eqnarray}
Crucially, this depends on the cut only through $p=\ell/N_s$, i.e.\ only through the size $\ell$ of $A$ and not through its geometry. At half filling of the sites, $\ell=N_s/2$ gives $p=\tfrac12$ for both the contiguous and the alternating cut, hence
\begin{equation}
S_{\mathrm{half}}=S_{\mathrm{even/odd}}=H\!\big(\Bin(N_s,\tfrac12)\big).
\end{equation}
The two entropies coincide not because of any lattice symmetry relating the cuts, but because the fully delocalized condensate cannot distinguish two subsets of equal size.

\subsection{Bounds and the logarithmic law}

Because the sum in \cref{eq:schmidtSFapp} runs over $n=0,\dots,N$, the Schmidt rank is at most $N+1$, giving the general bound
\begin{equation}
S_A\le\ln(N+1).
\end{equation}
For large $N$ the binomial distribution approaches a Gaussian of variance $Npq$, whose differential entropy yields
\begin{equation}
H\!\big(\Bin(N,p)\big)=\tfrac12\ln\!\big(2\pi e\,Npq\big)+O(N^{-1}).
\label{eq:gaussapp}
\end{equation}
At unit filling ($N=N_s$) and $p=q=\tfrac12$ this becomes
\begin{equation}
S_A\simeq\tfrac12\ln\!\Big(\tfrac{\pi e}{2}N_s\Big),
\label{eq:logleyapp}
\end{equation}
the logarithmic scaling characteristic of the SF phase \cite{Klich2009,Song2010}. Note that both cuts now grow with $N_s$, but only as $\tfrac12\ln N_s$, far slower than the MI volume law of the alternating cut.

\subsection{Single-body observables}

The same condensate state fixes the one-body observables used in the main text. The single-particle density matrix is uniform,
\begin{equation}
\rho_{ij}=\langle\bdag{i}\bop{j}\rangle=\frac{N}{N_s}\qquad\text{for all }i,j,
\end{equation}
so its largest eigenvalue is $N$ (the fully condensed mode) and the condensate fraction is $f_c=\lambda_{\max}/N=1$ \cite{PenroseOnsager1956,Yang}; the nearest-neighbor coherence saturates at $|\rho_{i,i+1}|=N/N_s=1$. Taking $\ell=1$ in \cref{eq:schmidtSFapp}, the occupation of a single site follows $\Bin(N,1/N_s)$, whose standard deviation is
\begin{equation}
\Delta(\nop{i})=\sqrt{N\,\tfrac{1}{N_s}\Big(1-\tfrac{1}{N_s}\Big)}
\;\xrightarrow{\,N=N_s\,}\;\sqrt{1-\tfrac{1}{N_s}},
\label{eq:dnSFapp}
\end{equation}
and the corresponding single-site entanglement entropy is $S_{\mathrm{site}}\to H\!\big(\Bin(N,1/N_s)\big)$.

\section{Details of the slave-boson method}
\label{app:sb}

This appendix collects the construction whose logic and central equations were given in the main text. We provide the optimization and constraint details of the mean-field state \cref{eq:sb-ansatz-main}, the explicit assembly of the BdG blocks in \cref{eq:sb-Hbdg-main}, and the leading-order reduction that connects the Gaussian entropy \cref{eq:S_gauss_main} to the MI perturbation theory.

\subsection{Mean-field state: optimization and constraints}
\label{app:sb-mf}

The local amplitudes of the two-sublattice state \cref{eq:sb-ansatz-main} minimize the energy per site at fixed mean density $\rho=N/N_s$, the Gutzwiller decoupling mean field for lattice bosons \cite{Sheshadri1993,vanOosten2001}. The minimization is performed by SLSQP subject to the self-consistent mean-field constraints (unit norm on each sublattice and fixed $\rho$). To set up the fluctuations, on each sublattice a local orthonormal basis is constructed whose first vector is the optimized state $\ket{\alpha}$; the remaining $n_{\max}$ excited states play the role of fluctuation bosons $\hat\gamma_{r,\alpha}$ ($\alpha=1,\dots,n_{\max}$), subject to the hard-core constraint $\sum_\alpha\hat\gamma^\dagger_{r,\alpha}\hat\gamma_{r,\alpha}=1$ that keeps a single physical state per site.

\subsection{Assembly of the BdG blocks}
\label{app:sb-bdg}

Collecting the fluctuation operators into $\hat\Gamma_k=(\hat\gamma_k,\hat\gamma_{k+\pi},\hat\gamma^\dagger_{ k},\hat\gamma^\dagger_{-k-\pi})$ brings the quadratic Hamiltonian to the block form \cref{eq:sb-Hbdg-main}. The blocks $A_k,B_k$ are assembled from the local matrix elements of $\hat b$, $\hat b^\dagger$ and $\hat n$ in the optimized basis and from the mean-field eigenenergies (the explicit even/odd, hopping, and long-range contributions are those of Ref.~\cite{Sharma2022}, adapted here to $d=1$ with structure factor $f_k=2\cos k$ and folding vector $\pi$). The matrix is diagonalized with the bosonic metric $\sigma=\mathrm{diag}(\mathbb I,-\mathbb I)$, i.e.\ by a paraunitary (Bogoliubov) transformation $T_k$ obeying $T_k^\dagger\sigma T_k=\sigma$ and $T_k^\dagger\mathcal H_{\mathrm{BdG}}(k)T_k=\mathrm{diag}(\omega_k,\omega_k)$, resulting in the physical frequencies $\omega_{k,p}\ge0$, the mean-field excitation spectrum of the lattice bosons \cite{Huber2007}. From the same transformation one obtains the correlation matrix $\mathcal C(k)$ of the ground state, which is then unfolded through \cref{eq:sb-unfold-main} and restricted to the sites of $A$ to give the reduced $\mathcal C_A$ entering \cref{eq:S_gauss_main}.

\subsection{Dominant-order and MI agreement}
\label{app:sb-mi}

In the MI phase there is no superfluid component ($\phi_{e,o}\approx0$) and all of $S_A$ comes from doublon-hole fluctuations, precisely those generated by the perturbation theory of the main text. The agreement is excellent over the entire $t/U$ range of this regime (\cref{tab:cmp_pert,tab:cmp_ratio}): the three methods coincide, and in particular the SB reproduces both the absolute value and the $N_s$-dependence of the ratio $S_{\mathrm{even/odd}}/S_{\mathrm{half}}\to N_s/2$.

This follows from the leading-order reduction \cref{eq:sb-leading-main} of the Gaussian entropy at small occupancies $n_m\ll1$, which matches the master formula \cref{eq:master} under $n_m\equiv\lambda_m=\varepsilon^2 c_m$. The entanglement-mode occupancies of the SB are then, mode for mode, the eigenvalues $\{c_k\}$ of the boundary matrix $CC^{\mathsf T}$ computed analytically in \cref{eq:Smitad,eq:Salt,eq:Ssitio}; the two constructions reconstruct the same spectrum. As noted in the main text, the mean field underestimates the MI-SF critical point, $(t/U)_c^{\mathrm{MF}}=\tfrac12(3-2\sqrt2)\approx0.086$, below the exact 1D value $(t/U)_c=0.305$ \cite{Ejima2011,Kuhner2000}; this limitation is well-known for mean-field theory and does not affect the MI or SF asymptotics.

\clearpage

\section{Supplementary figures}
\label{app:figures}

\begin{figure*}[h]
\centering
\includegraphics[width=0.9\textwidth]{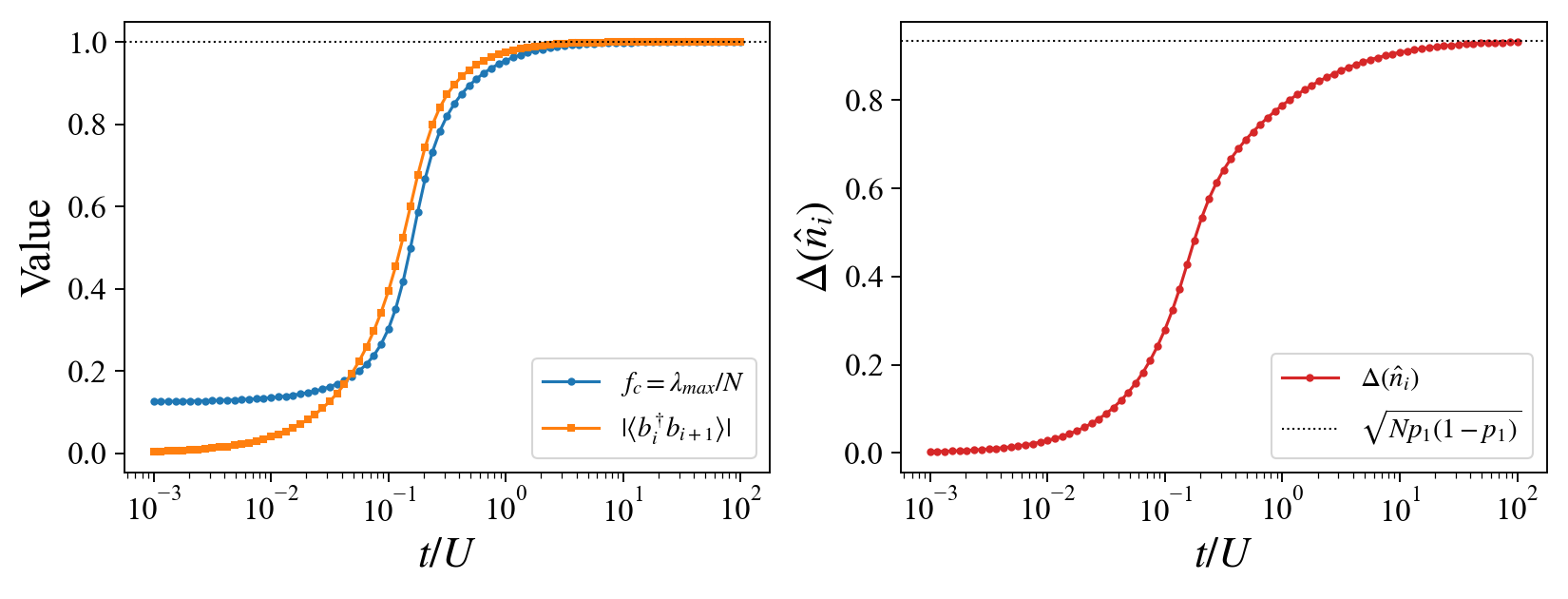}
\caption{Observables for $t/U$ from $10^{-4}$ to $10^{4}$, $N_s=N=8$. Left: condensate fraction and nearest-neighbor coherence. Right: on-site fluctuations.}
\label{fig:obscompleto}
\end{figure*}

\begin{figure}[h]
\centering
\includegraphics[width=0.6\columnwidth]{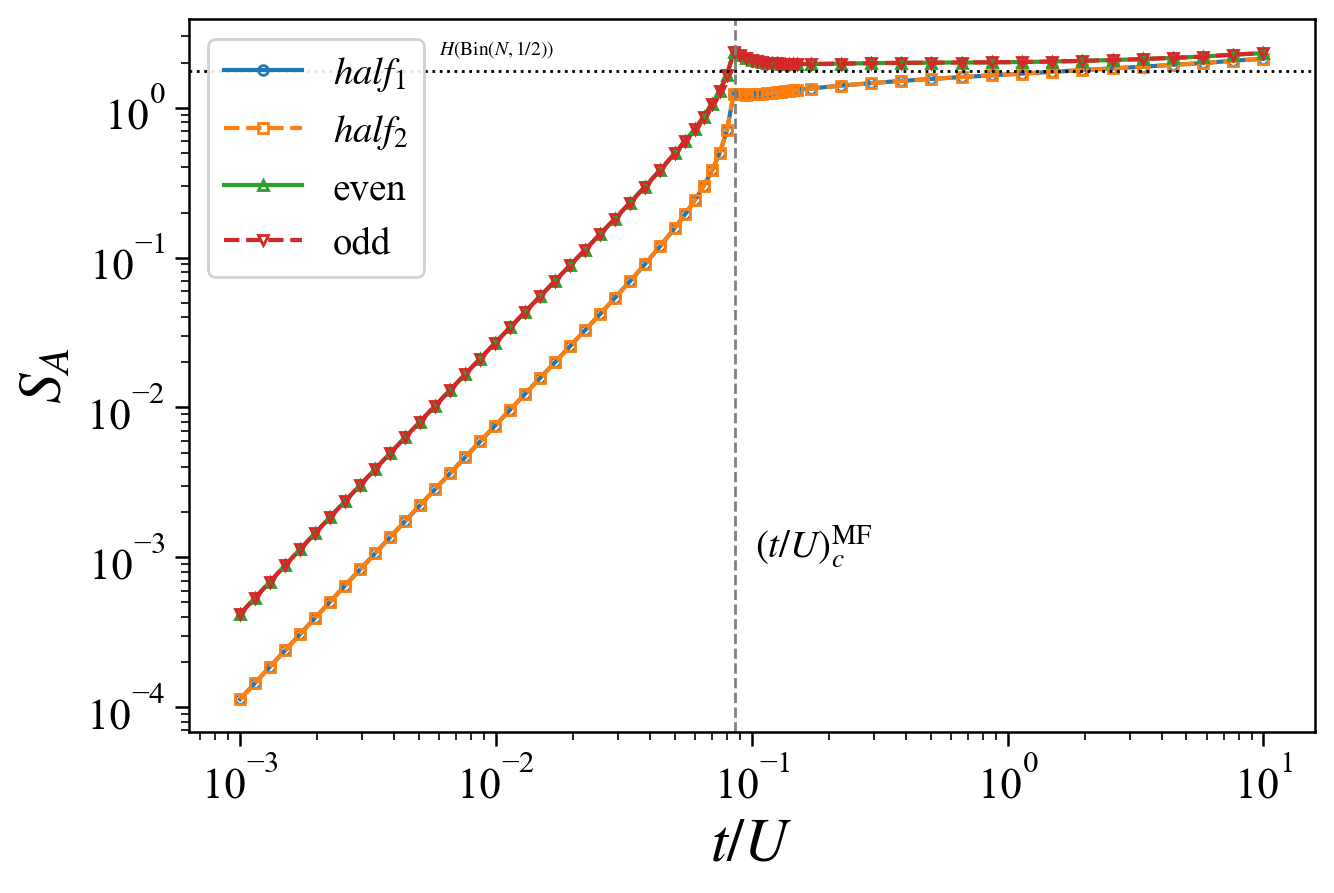}
\caption{The four bipartitions over the full range of $t/U$. The two exact equalities from the main texthold throughout. The two independent curves converge to the common value $H(\Bin(N,1/2))=1.6174$ for $N=8$.}
\label{fig:entcompleto}
\end{figure}

\begin{figure*}[t]
\centering
\includegraphics[width=0.9\textwidth]{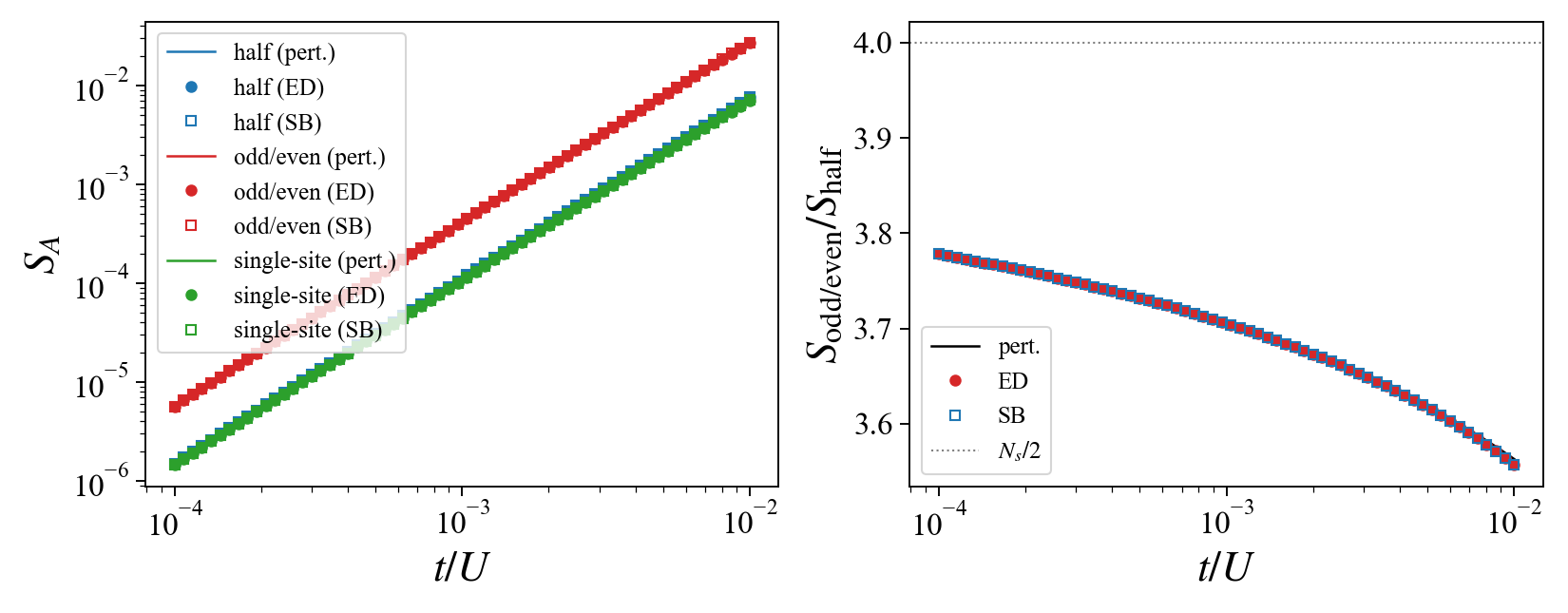}
\caption{MI regime, $N_s=N=8$. Comparison of exact diagonalization (ED), slave bosons (SB), and perturbation theory \cref{eq:master}. Left: $S_A$ for the half, even/odd, and single-site cuts. Right: ratio $S_{\mathrm{even/odd}}/S_{\mathrm{half}}$, tending to $N_s/2$.}
\label{fig:cmp_pert}
\end{figure*}

\begin{figure*}[t]
\centering
\includegraphics[width=0.9\textwidth]{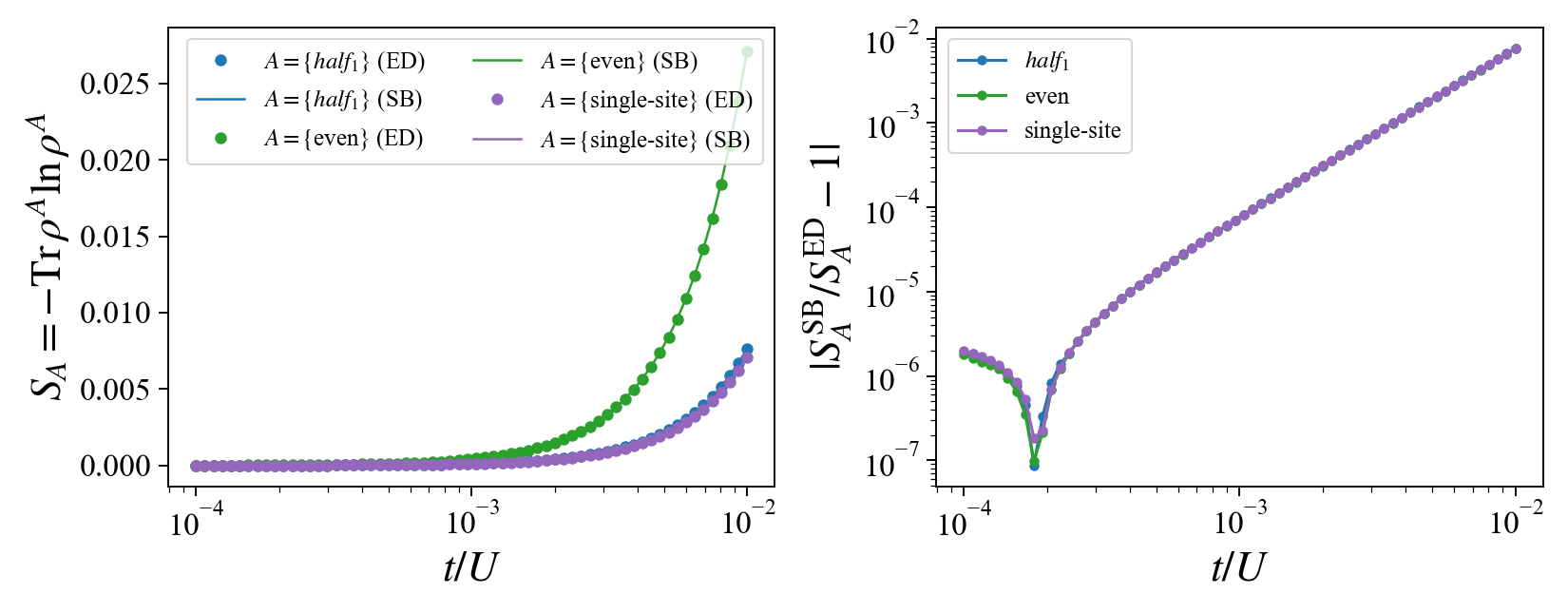}
\caption{MI regime. Left: $S_A$ from ED (symbols) and SB (lines) for the three cuts. Right: relative error $|S_A^{\mathrm{SB}}/S_A^{\mathrm{ED}}-1|$, small throughout the MI phase and growing slowly toward the transition, $N_s=N=8$.}
\label{fig:cmp_bipmi}
\end{figure*}
\clearpage
\section{Supplementary tables}
\label{app:tables}

\begin{table}[h!]
\centering
\caption{Ratio $S_{\mathrm{even/odd}}/S_{\mathrm{half}}$ in the MI regime ($N_s=N=8$): exact diagonalization (ED), slave bosons (SB), perturbation theory (pert.) \cref{eq:master}, and DMRG. All tend to $N_s/2=4$ in the deep MI.}
\label{tab:cmp_pert_ratio}
\begin{ruledtabular}
\begin{tabular}{lcccc}
$t/U$ & ED & SB & pert. & DMRG \\
\colrule
$1.00\times 10^{-2}$ & $3.5570$ & $3.5566$ & $3.5630$ & $3.5570$ \\
$1.79\times 10^{-2}$ & $3.4830$ & $3.4811$ & $3.5018$ & $3.4830$ \\
$3.22\times 10^{-2}$ & $3.3623$ & $3.3511$ & $3.4206$ & $3.3623$ \\
$5.78\times 10^{-2}$ & $3.1290$ & $3.0314$ & $3.3078$ & $3.1290$ \\
$8.96\times 10^{-2}$ & $2.7783$ & $1.8310$ & $3.1895$ & $2.7783$ \\
\end{tabular}
\end{ruledtabular}
\end{table}
\begin{table}[h!]
\centering
\caption{$S_A$ for the SF regime with $N_s=N=8$ ($t/U=100$), comparing simulation methods vs theory.}
\label{tab:sf}
\begin{ruledtabular}
\begin{tabular}{lcccc}
Quantity & ED & SB & DMRG & theoretical\\
\colrule
$S_{\mathrm{half}}$      & $1.7572$ & $2.6589$ & $1.6948$ & $H(\Bin(8,\tfrac12))=1.7635$\\
$S_{\mathrm{even/odd}}$  & $1.7627$ & $2.8319$ & $1.7873$ & $H(\Bin(8,\tfrac12))=1.7635$\\
$S_{\mathrm{site}}$      & $1.2681$ & $2.3003$ & ---        & $H(\Bin(8,\tfrac18))=1.2700$\\
\colrule
$\Delta(\nop{i})$    &   &$0.9128$ &$0.8124$ & $0.9128$\\
$f_c$                & &$1.0000$ & $0.9937$& $1$\\
$|\rho_{i,i+1}|$   &  & $1.0000$  &$1.0000$  &  $1$\\
\end{tabular}
\end{ruledtabular}
\end{table}

\begin{table}[h!]
\centering
\small
\caption{Scaling of the normalized ratio $(S_{\text{even/odd}}/S_{\text{half}})/N_s$ in the MI regime ($t/U=10^{-3}$).}
\label{tab:esc_mi}
\begin{ruledtabular}
\begin{tabular}{r ccc}
$N_s$ & ED & SB & DMRG\\
\colrule
$4$  & 0.4509 & 0.4509 & --      \\
$6$  & 0.4672 & 0.4672 & 0.4672\\
$8$  & 0.4631 & 0.4631 & 0.4631\\
$10$ & 0.4651 & 0.4651 & 0.4651\\
$12$ & --       & 0.4641 & 0.4641\\
$16$ & --       & 0.4644 & 0.4643\\
$32$ & --       & 0.4645 & 0.4645\\
$64$ & --       & 0.4645 & 0.4644\\
\end{tabular}
\end{ruledtabular}
\end{table}

\begin{table}[h!]
\centering
\small
\caption{Scaling of the ratio $S_{\text{even/odd}}/S_{\text{half}}$ in the SF regime ($t/U=10^{2}$).}
\label{tab:esc_sf}
\begin{ruledtabular}
\begin{tabular}{r ccc}
$N_s$ & ED & SB & DMRG\\
\colrule
$4$  & 1.0005 & 1.0111 & --      \\
$6$  & 1.0016 & 1.0050 & 1.0215 \\
$8$  & 1.0030 & 1.0097 & 1.0545\\
$10$ & 1.0048 & 1.0159 & 1.0935 \\
$12$ & --       & 1.0235 & 1.1354\\
$16$ & --       & 1.0434 & 1.2243 \\
$32$ & --       & 1.1703 & 1.5643 \\
$64$ & --       & 1.5437 & 1.6164\\
\end{tabular}
\end{ruledtabular}
\end{table}

\end{document}